\documentclass[11pt,a4paper]{article}

\usepackage[a4paper,margin=2.5cm]{geometry}
\usepackage[utf8]{inputenc}
\usepackage[T1]{fontenc}
\usepackage{bbm,braket,amsmath,amsthm,amsfonts,float,tikz,algorithm2e,thm-restate}
\usepackage[english]{babel}
\usepackage[colorlinks=true,linkcolor=blue,urlcolor=blue,citecolor=blue]{hyperref}
\usepackage[capitalize,nameinlink]{cleveref}
\usepackage{authblk}

\usepackage{thmtools}
\declaretheorem[style=plain,numberwithin=section]{theorem}
\declaretheorem[style=plain,numberlike=theorem]{lemma,corollary}

\declaretheorem[style=definition,numberlike=theorem]{definition}

\newcommand{\BQL}{\mathsf{BQL}}
\newcommand{\BPL}{\mathsf{BPL}}

\newcommand{\C}{\mathbbm C}

\begin{document}

\title{Directed $st$-connectivity with few paths is in quantum logspace}
\author[1]{Simon Apers}
\author[1]{Roman Edenhofer}
\affil[1]{Université Paris Cité, CNRS, IRIF, Paris, France}
\date{}
\maketitle

%%%%%%%%%%%%%%%%%%%% ABSTRACT: %%%%%%%%%%%%%%%%%%%%
\begin{abstract}
We present a $\mathsf{BQSPACE}(O(\log n))$-procedure to count $st$-paths on directed graphs for which we are promised that there are at most polynomially many paths starting in $s$ and polynomially many paths ending in $t$.
For comparison, the best known classical upper bound in this case just to decide $st$-connectivity is $\mathsf{DSPACE}(O(\log^2 n/ \log \log n))$.
The result establishes a new relationship between~$\mathsf{BQL}$ and unambiguity and fewness subclasses of $\mathsf{NL}$.
Further, we also show how to \emph{recognize} directed graphs with at most polynomially many paths between any two nodes in $\mathsf{BQSPACE}(O(\log n))$. This yields the first natural candidate for a language separating $\mathsf{BQL}$ from $\mathsf{L}$ and~$\mathsf{BPL}$.
Until now, all candidates potentially separating these classes were inherently promise problems.
\end{abstract}

%%%%%%%%%%%%%%%%%%%% INTRODUCTION: %%%%%%%%%%%%%%%%%%%%

\section{Introduction and summary}
\label{sec: intro}

Graph connectivity is a central problem in computational complexity theory, and it is of particular importance in the space-bounded setting.
Given a graph $G$ and two vertices $s$ and~$t$, the task is to decide whether there is a path from $s$ to $t$.
For undirected graphs the problem is denoted as~$\mathrm{USTCON}$. Aleliunas, Karp, Lipton, Lovász and Rackoff \cite{AKL+79} showed that doing a random walk for a polynomial number of steps can solve it in \emph{randomized logspace},~$\mathsf{RL}$. After a long line of work, Reingold \cite{Rei05} was able to derandomize the result and showed that the problem is already contained in \emph{deterministic logspace}, $\mathsf{L}$, and is in fact complete for that class.
In the directed graph setting the problem is denoted as $\mathrm{STCON}$, and it is complete for \emph{non-deterministic logspace}, $\mathsf{NL}$. The best known deterministic algorithm for $\mathrm{STCON}$ in terms of space complexity is due to Savitch \cite{Sav70} and runs in space $O(\log^2 n)$.
We have the following well-known inclusions
\begin{align*}
    \mathsf{L} \subseteq \mathsf{RL} \subseteq \mathsf{NL} \subseteq \mathsf{DET} \subseteq \mathsf{L}^2
\end{align*}
where $\mathsf{DET}$, introduced by Cook \cite{Coo85}, is the class of languages that are $\mathsf{NC}^1$ Turing reducible to the computation of the determinant of an integer matrix and $\mathsf{L}^2$ is the class of languages decidable in deterministic space $O(\log^2 n)$.
We refer to \cite{AB06} for further details on reductions and basic complexity classes.

The most studied quantum space-bounded complexity class is \emph{bounded error quantum logspace}, denoted~$\mathsf{BQL}$.
This is the class of languages decided in~$\mathsf{BQSPACE}(O(\log n))$, i.e., decided by a quantum Turing machine with error $1/3$ running in space $O(\log n)$ and time~$2^{O(\log n)}$. The quantum Turing machine has read-only access to the input on a classical tape and can write on a uni-directional output tape which does not count towards the space complexity.
The class~$\mathsf{BQL}$ lies in between~$\mathsf{RL}$ and~$\mathsf{DET}$.
In fact, it was recently shown by Fefferman and Remscrim~\cite{FR21} that certain restricted versions of the standard $\mathsf{DET}$-complete matrix problems are complete for~$\mathsf{prBQL}$, where~$\mathsf{prBQL}$ is the promise version of $\mathsf{BQL}$, i.e. the class of promise problems decided in~$\mathsf{BQSPACE}(O(\log n))$.
This extended the earlier work of Ta-Shma~\cite{TS13}, who showed how to invert well-conditioned matrices in $\mathsf{BQSPACE}(O(\log n))$ building on the original idea of Harrow, Hassidim and Lloyd~\cite{HHL09}.
We restate two of Ta-Shma's main results:
\begin{enumerate}
    \item (Compare \cite[Theorem 5.2]{TS13})
        Given a matrix $M \in \C^{n \times n}$ with $\|M\|_2 \leq \mathrm{poly}(n)$\footnote{Throughout this paper, $\|\cdot\|_2$ shall always denote the $\ell_2$-norm of a vector or the spectral norm of a matrix.}, we can output all of its singular values and their respective multiplicities up to~$1/\mathrm{poly}(n)$ additive accuracy in $\mathsf{BQSPACE}(O(\log n))$. In particular, we can determine $\dim(\ker(M))$ if all non-zero singular values have inverse polynomial distance from zero.
    \item (Compare \cite[Theorem 1.1]{TS13}) 
        Given two indices $s,t\in [n]$ and a matrix $M\in\C^{n\times n}$ which is \emph{poly-conditioned}, by this we mean that its singular values satisfy
        \begin{align*}
            \mathrm{poly}(n)\geq \sigma_1(M) \geq ... \geq \sigma_n(M) \geq 1/\mathrm{poly}(n),
        \end{align*}
        we can estimate $M^{-1}(s,t)$ up to $1/\mathrm{poly}(n)$ additive accuracy in $\mathsf{BQSPACE}(O(\log n))$.
\end{enumerate}

The first result above directly implies a $\mathsf{BQL}$-procedure for deciding $\mathrm{USTCON}$ (alternatively, this clearly follows from the containment $\mathsf{RL} \subseteq \mathsf{BQL}$).
To see this, note that (i)~$\mathrm{USTCON}$ can be reduced to counting the number of connected components, (ii) the dimension of the kernel of the random walk Laplacian $I-P$ is equal to the number of connected components, and (iii) for undirected graphs the spectral gap of $I-P$, that is its smallest non-zero eigenvalue, is inverse polynomially bounded from zero.
Here $I$ is the identity and $P$ is the transition matrix of a random walk.
This ties to the fact that a random walk on an undirected graph takes polynomial time to traverse the graph (see e.g.~\cite[Chapter 12]{LPW08}).
Unfortunately, for directed graphs the situation is more complicated. Importantly, the smallest non-zero singular value of the random walk Laplacian can be inverse exponentially small, and similarly the time it takes a random walk to find connected nodes can be exponential.
Hence, it is not obvious how Ta-Shma's results should be of any help in this setting.

\subsubsection*{Counting few paths}
Somewhat surprisingly, we show that by analyzing a different matrix which we call the \emph{counting Laplacian} $L = I-A$ we can use Ta-Shma's second result to solve instances of $\mathrm{STCON}$ in~$\mathsf{BQSPACE}(O(\log n))$ that seem hard classically. Here $A$ denotes the adjacency matrix of a graph.
While the random walk Laplacian is spectrally well-behaved if a random walk is efficient, we find that the counting Laplacian is spectrally well-behaved if the underlying graph is acyclic and contains only few paths.
We remark that the number of paths in a graph can be totally unrelated to the success probability of a random walk finding specific nodes.
Even if there exist very few paths, it can happen that a random walk has an extreme bias to only pick paths we are not interested in.
Consider for instance the following graph:

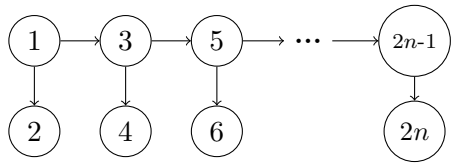
\begin{figure}[H]
    \centering
    \begin{tikzpicture}[scale=.8]
        \node[shape=circle,draw=black] (1) at (0,0) {$1$};
        \node[shape=circle,draw=black] (2) at (0,-1.5) {$2$};
        \node[shape=circle,draw=black] (3) at (1.5,0) {$3$};
        \node[shape=circle,draw=black] (4) at (1.5,-1.5) {$4$};
        \node[shape=circle,draw=black] (5) at (3,0) {$5$};
        \node[shape=circle,draw=black] (6) at (3,-1.5) {$6$};
        \node[] (dots) at (4.5,0) {\textbf{...}};
        \node[shape=circle,draw=black] (7) at (6.25,0) {\scriptsize $2n$-$1$};
        \node[shape=circle,draw=black] (8) at (6.25,-1.5) {\small $2n$};
    
        \path [->](1) edge node[left] {} (2);
        \path [->](1) edge node[left] {} (3);
        \path [->](3) edge node[left] {} (4);
        \path [->](3) edge node[left] {} (5);
        \path [->](5) edge node[left] {} (6);
        \path [->](5) edge node[left] {} (dots);
        \path [->](dots) edge node[left] {} (7);
        \path [->](7) edge node[left] {} (8);
    \end{tikzpicture}
    \label{exponential abort}
\end{figure}

\noindent The number of paths between any two nodes is at most one. Nonetheless, a random walk starting at node $1$ only has probability $1/2^{n-1}$ to reach node $2n$.

Given a directed graph $G=(V,E)$ with nodes $i,j\in V$ let us denote by $N(i,j)$ the number of paths from $i$ to $j$. By a path from $i$ to $j$ we mean a sequence of edges $(e_1,...,e_k)$ which joins a sequence of nodes $(v_0,...,v_k)$ such that $v_0=i$, $v_k=j$ and $e_i = (v_{i-1},v_i)$ for all~$i\in [k]$.\footnote{Some authors call this a \emph{walk} and disallow a path to visit any vertex more than once. We follow~\cite{AL98,KKR08}~in our definition. By convention we also count the empty sequence as a path from any node to itself such that~$N(i,i)\geq 1$ for all $i\in V$.}
As our first result, and as a primer for the rest of the paper, we show that we can count the number of~$st$-paths on directed graphs for which there are at most polynomially many paths between any two nodes in~$\mathsf{BQSPACE}(O(\log n))$. In particular, this allows to decide~$\mathrm{STCON}$ on such graphs.

\begin{theorem}\label{thrm: counting paths on StFewL-graphs}
    Fix a polynomial $p:\mathbb{N} \rightarrow \mathbb{N}$.
    Let $G$ be a directed graph with $|V(G)|=n$ nodes such that
    \begin{itemize}
        \item $\forall i,j \in V(G): N(i,j)\leq p(n)$.
    \end{itemize}
    There is an algorithm running in $\mathsf{BQSPACE}(O(\log n))$ that, given access to the adjacency matrix~$A$ of $G$ and $s,t\in V(G)$, returns the number of paths from $s$ to $t$.
\end{theorem}
\begin{proof}
    First of all, observe that the finite path bound between any two nodes implies that the graph is in particular acyclic.
    Second, note that $A^k(i,j)$ is equal to the number of paths of length $k$ from node $i$ to node $j$.
    Both in mind, it follows that $A^n = 0$, i.e. $A$ is nilpotent.
    As a consequence we obtain that the inverse of the counting Laplacian exists and is equal to 
    \begin{equation*}
        L^{-1} = (I-A)^{-1} = I + A + A^2 + ... + A^{n-1}.
    \end{equation*}
    Clearly, its entries simply count the total number of paths from $i$ to $j$, i.e.~$L^{-1}(i,j)~=~N(i,j)$.
    Crucially, we find that the polynomial bound on the number of paths ensures that the matrix is poly-conditioned.
    Observe that the $\max$-norm of a matrix, i.e. its maximum entry in absolute value, can be used to bound its spectral norm. More precisely, we have that~$\|M\|_2~\leq~n~\cdot~\|M\|_{\max}$ for any matrix $M$.
    Using that $\|L^{-1}\|_{\max} = \max_{i,j\in [n]} |L^{-1}(i,j)| \leq p(n)$ and $\| L \|_{\max} = 1$, we thus find bounds for the largest and smallest singular value of~$L$,
    \begin{align*}
        n \geq \sigma_1(L) 
        \quad \text{and}\quad \sigma_n(L) = \frac{1}{\|L^{-1}\|_2} \geq \frac{1}{n\cdot\|L^{-1}\|_{\max}} \geq \frac{1}{n\cdot p(n)}.
    \end{align*}
    Hence, the counting Laplacian is poly-conditioned and we can apply Ta-Shma's matrix inversion algorithm, the second restated result above, to approximate entry $L^{-1}(s,t) = N(s,t)$ up to additive error $1/3$ in $\mathsf{BQSPACE}(O(\log n))$. Rounding to the closest integer gives the number of paths from $s$ to $t$.
\end{proof}

The proof of \cref{thrm: counting paths on StFewL-graphs} uses a reduction of summing matrix powers to matrix inversion.
The idea for this is well-known, and appeared in this context in early work by Cook \cite{Coo85}, and more recently by Fefferman and Remscrim \cite{FR21}. However, they did not study the implications to the space complexity of $\mathrm{STCON}$.

In \cref{sec:4} we push the algorithmic idea of \cref{thrm: counting paths on StFewL-graphs} further by doing a more fine-grained analysis of the counting Laplacian's singular values and singular vectors.
We show that we can already count $st$-paths in $\mathsf{BQSPACE}(O(\log n))$ on directed graphs for which only the number of paths starting from $s$ and the number of paths ending in~$t$ are polynomially bounded.
We state the formal result here.

\begin{restatable}{theorem}{main}
    \label{main}
    Fix a polynomial $p: \mathbb{N} \rightarrow \mathbb{N}$. Let $G$ be a directed graph with $|V(G)|=n$ nodes such that 
    \begin{itemize}
        \item $\forall j \in V(G) : N(s,j) \leq p(n)$ and $N(j,t) \leq p(n)$.
    \end{itemize}
    There is an algorithm running in $\mathsf{BQSPACE}(O(\log n))$ that, given access to the adjacency matrix~$A$ of $G$ and $s,t\in V(G)$, returns the number of paths from $s$ to $t$.
\end{restatable}

Classically (deterministically or randomly), the best known space bound just to decide $\mathrm{STCON}$ on graphs promised to satisfy a polynomial bound on the number of paths between any two nodes as in \cref{thrm: counting paths on StFewL-graphs} or only on the number of paths starting from $s$ and on the number of paths ending in $t$ as in \cref{main} is $\mathsf{DSPACE}(O(\log^2(n)/\log\log(n)))$~\cite{AL98,GSTV11}.
Alternatively, as noticed in \cite{BJLR91,Lan97} we can also solve such $\mathrm{STCON}$ instances simultaneously in deterministic polynomial time and space $O(\log^2 n)$.
We elaborate on these bounds in the next section.

\subsubsection*{Unambiguity, Fewness and a language in $\mathsf{BQL}$ (maybe) not in $\mathsf{L}$}

Previous works studied restrictions on the path count in connection to ``unambiguity'' and ``fewness'' of configuration graphs of space-bounded Turing machines.
These notions were introduced to interpolate between determinism and non-determinism, and gives rise to complexity classes between $\mathsf{L}$ and $\mathsf{NL}$. We formally introduce these classes in \cref{section2}.
As a direct consequence of \cref{thrm: counting paths on StFewL-graphs} we obtain:
\begin{restatable}{corollary}{stfewl}
\label{stfewl}
    $\mathsf{StrongFewL} \subseteq \mathsf{BQL}$
\end{restatable}
\noindent where $\mathsf{StrongFewL}$ denotes the class of languages which are decidable by an $\mathsf{NL}$-Turing machine whose configuration graphs for all inputs are \emph{strongly-few}, that is graphs for which there are at most polynomially many paths between any two nodes.
By some preprocessing in the algorithm of \cref{thrm: counting paths on StFewL-graphs} we can actually check in $\mathsf{BQSPACE}(O(\log n))$ whether a given graph is strongly-few.\footnote{Unfortunately, we see no way to do the same for the weaker promise of \cref{main}.}
We obtain the subsequent language containment (notably, the containment of this language in $\mathsf{StrongFewL}$ is not known).

\begin{restatable}{theorem}{STCONsf}\label{thrm: STCONsf}
\label{STCONsf}
    The language
    \begin{align*}
        \mathrm{STCON}_{\mathrm{sf}}
            = \{ \langle G,s,t,1^k\rangle ~|~\forall i,j\in V(G): N(i,j)\leq k \text{ and } N(s,t) \geq 1 \}
    \end{align*}
    is contained in $\mathsf{BQL}$.
\end{restatable}

This seems to be the first concrete example of a \emph{language} in $\mathsf{BQL}$ not known to lie in~$\mathsf{L}$ or~$\mathsf{BPL}$.
As far as we are aware, before this work only \emph{promise problems} were known that lie in~$\mathsf{prBQL}$ but which are potentially not contained in~$\mathsf{prBPL}$~\cite{TS13,FL18,FR21,GLW23}.

We now summarize briefly what is known classically.
A language similar to $\mathrm{STCON}_{\mathrm{sf}}$, namely 
\begin{align*}
    \mathrm{STCON}_{\mathrm{ru}} = \{ \langle G,s,t\rangle ~|~ \forall j \in V(G): N(s,j) \leq 1 \text{ and } N(s,t) = 1 \}
\end{align*} has been studied before.
It was first introduced by Lange \cite{Lan97} who showed that it is complete for $\mathsf{ReachUL}$, which is the class of languages decidable by an $\mathsf{NL}$-Turing machine whose configuration graphs for all inputs are \emph{reach-unambiguous}, meaning that there is a unique computation path from the start configuration to any reachable configuration.\footnote{Note that the $\mathsf{ReachUL}$-hardness of $\mathrm{STCON}_{\mathrm{ru}}$ is trivial but the completeness is not. This is because the uniqueness in the definition of the complexity class is used as a restriction on the machine, while it is used as an acceptance criterion in the definition of the language.}
Further, it was noticed in~\cite{Lan97, BJLR91} that $\mathsf{ReachUL}$ is contained in~$\mathsf{SC}^2$, which is the class of languages decidable simulatenously in deterministic polynomial time and space~$O(\log^2 n)$.
Additionally, Allender and Lange~\cite{AL98} found that $\mathrm{STCON}$ on graphs promised to be reach-unambiguous is solvable in deterministic space $O(\log^2(n)/ \log\log(n))$ implying~$\mathsf{ReachUL}~\subseteq~\mathsf{DSPACE}(O(\log^2 (n)/ \log \log(n)))$.
In particular, these results put~$\mathrm{STCON}_{\mathrm{ru}}$ in~$\mathsf{SC}^2$ and in~$\mathsf{DSPACE}(O(\log^2(n)/ \log\log(n)))$.
Also, more recently Garvin, Stolee, Tewari and Vinodchandran \cite{GSTV11} proved that $\mathsf{ReachUL}$ is equal to $\mathsf{ReachFewL}$, where $\mathsf{ReachFewL}$ is the class of languages decidable by an $\mathsf{NL}$-Turing machine whose configuration graphs for all inputs are \emph{reach-few}, i.e. they have at most polynomially many computation paths from its start configuration to any reachable configuration.
We thus have
\begin{align*}
    \mathsf{StrongFewL} \subseteq \mathsf{ReachFewL} = \mathsf{ReachUL} \subseteq \mathsf{SC}^2, \text{ }\mathsf{DSPACE}(O(\log^2(n) / \log \log(n))).
\end{align*}
Notably, while \cite{GSTV11} implies a procedure to solve $\mathrm{STCON}$ on graphs promised to be reach-few in $\mathsf{DSPACE}(O(\log^2(n) / \log \log(n)))$ (in particular, this includes the graphs from \cref{thrm: counting paths on StFewL-graphs,main}), this does not allow to verify whether a general graph satisfies this promise.
Indeed, none of the upper bounds for $\mathrm{STCON}_{\mathrm{ru}}$ directly carry over to $\mathrm{STCON}_{\mathrm{sf}}$, for which the best classical space bound to date seems to be $O(\log^2 n)$.

\subsubsection*{Conclusion and open questions}

Summarizing, we show that in quantum logspace we can count $st$-paths on graphs for which even solving $st$-connectivity is not known to be possible in classical (deterministic or randomized) logspace.
Further, we obtain the first language in $\mathsf{BQL}$ not known to lie in $\mathsf{L}$ or~$\mathsf{BPL}$.
Our work also yields a number of open questions.

An obvious first question is whether $\mathrm{STCON}_{\mathrm{sf}}$ really separates $\mathsf{BQL}$ and $\mathsf{BPL}$.
A first step towards answering this might be tackling the related question, whether we can carry over some of the known upper bounds for $\mathrm{STCON}_{\mathrm{ru}}$ to $\mathrm{STCON}_{\mathrm{sf}}$. That is, is $\mathrm{STCON}_{\mathrm{sf}}$ also contained in~$\mathsf{SC}^2$ or in~$\mathsf{DSPACE}(O(\log^2(n)/\log\log(n)))$?
Note that the known $\mathsf{prBQL}$-complete problems are not known to satisfy these upper bounds \cite{TS13,FL18,FR21,GLW23}.

Further, we wonder whether it is possible to improve the $O(\log^2(n)/ \log \log(n))$ space bound by Allender and Lange.
In fact, Allender recently asked this in an article \cite{All23} in which he reflects on some open problems he encountered throughout his career. Allender suspects that for the restricted case of \emph{strongly-unambiguous} graphs, also called \emph{mangroves}, for which the number of paths between any two nodes is bounded by one, there should exist an algorithm deciding $\mathrm{STCON}$ running in deterministic space $O(\log n)$. He even offers a~\$1000 reward for any improvement of their space bound, already for this restricted case. A dequantization of our results would thus yield some good pocket money.\footnote{Even if the dequantization were randomized, i.e. in $\mathsf{BPSPACE}(O(\log n))$, it would still imply a $\mathsf{DSPACE}(O(\log^{3/2} n))$ bound by the inclusion $\mathsf{BPSPACE}(O(\log n)) \subseteq \mathsf{DSPACE}(O(\log^{3/2} n))$ due to Saks and Zhou~\cite{SZ99} which has recently been slightly improved on by Hoza \cite{Hoz21}.}

Another natural question is whether in $\mathsf{BQSPACE}(O(\log n))$ we can solve $\mathrm{STCON}$ on a larger class of graphs such as reach-few ones, where only the number of paths from $s$ is polynomially bounded.
This relaxes our promise in \cref{main}.
Unfortunately, our current approach of using an effective pseudoinverse seems to require our stronger promise.
More generally, we feel that a better understanding of the singular values and vectors of directed graph matrices will yield further insights into the utility of quantum space-bounded computation for solving problems on directed graphs.

We further believe it would be interesting to investigate whether there is a quantum algorithm running simultaneously in polynomial time and truly sublinear space~$O(n^{1-\varepsilon})$ for some~$\varepsilon>0$ deciding $\mathrm{STCON}$ on general graphs. Note that the best known classical algorithm for deciding $\mathrm{STCON}$ on general graphs and running in polynomial time needs almost linear space~$\Omega(n/2^{\sqrt{\log n}})$ as described in \cite{BBRS98}.

Finally, the link between poly-conditionedness and bounds on the path count raises the question of whether some variation of $\mathrm{STCON}$ could be proven complete for $\mathsf{BQL}$.

%%%%%%%%%%%%%%%%%%%% SPACE-BOUNDED COMPUTATION: %%%%%%%%%%%%%%%%%%%%

\section{Space-bounded computation}
\label{section2}

In \cref{sec:TMs} we introduce the Turing machine model of space-bounded computation and define the most important appearing complexity classes.
We mainly follow the definitions from Ta-Shma~\cite{TS13} and refer to~\cite[Section 2.2]{FR21} for a discussion on the equivalence of the quantum Turing machine model and the quantum circuit model.
In \cref{sec:complete} we mention some complete problems of promise versions of the defined complexity classes.

\subsection{Turing machines} \label{sec:TMs}

A \emph{deterministic} space-bounded Turing machine (DTM) acts according to a transition function~$\delta$ on three semi-infinite tapes: A read-only tape where the input is stored, a read-and-write work tape and a uni-directional write-only tape for the output.
The TMs computation time is defined as the number of transition steps it performs on an input, and its computation space is the number of used cells on the work tape, i.e. we do not count the number of cells on the input or output tape towards its computation space.
DTMs with space-bound $s(n)$ for inputs of length $n$ give rise to~$\mathsf{DSPACE}(s(n))$.
We define $\mathsf{L}$ as the class of languages decided in $\mathsf{DSPACE}(O(\log n))$ and $\mathsf{L}^2$ as the class of languages decided in $\mathsf{DSPACE}(O(\log^2 n))$.

A \emph{non-deterministic} Turing machine (NTM) is similar to a DTM except that it has two transition functions $\delta_0$ and $\delta_1$. At each step in time the machine non-deterministically chooses to apply either one of the two. 
It is said to accept an input if there is a sequence of these choices so that it reaches an accepting configuration and it is said to reject the input if there is no such sequence of choices.
We obtain $\mathsf{NSPACE}(s(n))$. Further, $\mathsf{NL}$ is the class of languages decided in~$\mathsf{NSPACE}(O(\log n))$.

A \emph{probabilistic} space-bounded Turing machine (PTM) is again similar to a DTM but with the additional ability to toss random coins. This can be conveniently formulated by a fourth tape that is uni-directional, read-only and initialized with uniformly random bits at the start of the computation. It does not count towards the space.
A language is said to be decided in~$\mathsf{BPSPACE}_{a,b}(s(n))$ if there is a PTM running in space $s(n)$ and time\footnote{\label{time}The time bound does not follow from the space-bound and is equivalent to demanding that the TM absolutely halts for all possible assignments of the random coins tape.} $2^{O(s(n))}$ deciding it with completeness error $a\in[0,1]$ and soundness error $b\in[0,1]$, that is every input in the language is accepted with probability at least $a$ and every input not in the language is accepted with probability at most $b$. 
Also, we write~$\mathsf{BPSPACE}(s(n))$ for~$\mathsf{BPSPACE}_{\frac{1}{3},\frac{2}{3}}(s(n))$ and $\mathsf{RSPACE}(s(n))$ for~$\mathsf{BPSPACE}_{\frac{1}{2},0}(s(n))$. Further,~$\mathsf{BPL}$ is the class of languages decided in~$\mathsf{BPSPACE}(O(\log n))$ and~$\mathsf{RL}$ is the class of languages decided in~$\mathsf{RSPACE}(O(\log n))$.

A \emph{quantum} space-bounded Turing machine (QTM) is a DTM with a fourth tape for quantum operations instead of a random coins tape. 
The transition function $\delta$ is still classically described. The tape cells of the quantum tape are qubits and initialized in state~$\ket{0}$ at the start of the computation. There are two tape heads moving in both directions on the quantum tape. At each step during the computation, the machine can apply a gate from some universal gate set, say $\{\mathrm{HAD},\mathrm{T},\mathrm{CNOT}\}$, or perform a measurement to a projection in the standard basis to the qubits below its tape heads.
The measurement outcomes are communicated to the machine and can therefore control later steps of the computation. In particular, the machine can reset qubits to their initial state.
The used cells on the quantum tape count towards the computation space.
As before, we say a language is decided in~$\mathsf{BQSPACE}_{a,b}(s(n))$ if there is a QTM running in space~$s(n)$ and time $2^{O(s(n))}$ deciding it with completeness error $a$ and soundness error $b$. Also, we mean $\mathsf{BQSPACE}_{\frac{2}{3},\frac{1}{3}}(s(n))$ by~$\mathsf{BQSPACE}(s(n))$ if not mentioned otherwise. 
Finally,~$\mathsf{BQL}$ is the class of languages decided in~$\mathsf{BQSPACE}(O(\log n))$.
The particular choice of the universal gate set does not affect the resulting complexity class due to the space-efficient version of the Solovay-Kitaev theorem of van Melkebeek and Watson \cite{MW10}.
The same applies to disallowing intermediate measurements, thanks to the space-efficient ``deferred measurement principle'' proven by Fefferman and Remscrim \cite{FR21} building on the earlier work of Fefferman and Lin~\cite{FL18} (see also~\cite{GRZ21,GR21} for an alternative time- and space-efficient version of this principle which applies to a special type of intermediate measurements that cannot control later computations).
Also, the chosen success probability of $2/3$ can be amplified by sequentially repeating computations.

\subsection{Complete (promise) problems} \label{sec:complete}
In this work we only consider language classes. In particular, we present the first candidate to distinguish $\mathsf{BQL}$ and $\mathsf{BPL}$. Note that candidates for distinguishing the promise versions of these classes, $\mathsf{prBQL}$ and $\mathsf{prBPL}$, are well-known.
As mentioned in the introduction, Fefferman and Remscrim \cite{FR21} showed that well-conditioned promise versions of all the standard~$\mathsf{DET}$-complete matrix problems are complete for~$\mathsf{prBQL}$.
Let us highlight that matrix powering is one of these problems. Restricted to stochastic matrices it is easily seen to be complete for~$\mathsf{prBPL}$, and restricted to matrices for which the largest singular values of its powers grow at most polynomially, it is complete for~$\mathsf{prBQL}$.
Indeed, this seems to indicate that powering the adjacency matrix of a graph, to which our approach in \cref{thrm: counting paths on StFewL-graphs} essentially boils down to, rather than powering the corresponding random walk matrix, truly exploits a quantum advantage.
A similar distinction of $\mathsf{prBQL}$ and $\mathsf{prBPL}$ is expected from the approximation of the spectral gap of matrices. Doron, Sarid and Ta-Shma~\cite{DST17} showed that a promise decision version of this problem is $\mathsf{prBPL}$-complete for stochastic matrices while it is~$\mathsf{prBQL}$-complete for general Hermitian matrices.
More recently, Le Gall, Liu and Wang~\cite{GLW23} presented another group of~$\mathsf{prBQL}$-complete problems based on state testing.

%%%%%%%%%%%%%%%%%%%% UNAMBIGUITY AND FEWNESS: %%%%%%%%%%%%%%%%%%%%

\section{Fewness language in \textsf{BQL}}

In \cref{sec:unamb} we introduce the notions of unambiguity and fewness in space-bounded computation.
In \cref{sec:BQL-language} we prove \cref{STCONsf} presenting a language in $\mathsf{BQL}$ not known to lie in~$\mathsf{L}$ or $\mathsf{BPL}$.

\subsection{Unambiguity and fewness} \label{sec:unamb}

The computation of a Turing machine can be viewed as a directed graph on configurations, and certain restrictions on the Turing machine translate to natural restrictions on the corresponding configuration graph.
The notions of ``unambiguity'' and ``fewness'' of a Turing machine relate to the following graph-theoretic notions.
\begin{definition}\label{def:unambiguity and fewness}
    Let $G=(V,E)$ be a directed graph and let $k$ be an integer. Then $G$ is called
    \begin{itemize}
        \item $k$-\emph{unambiguous} with respect to nodes $s,t\in V$ if $N(s,t) \leq k$,
        \item $k$-\emph{reach-unambiguous} with respect to node $s\in V$ if for all $j \in V,$ $N(s,j) \leq k$,
        \item $k$-\emph{strongly unambiguous} if for all $i,j \in V,$ $N(i,j) \leq k$.
    \end{itemize}
In the case of $k=1$, we simply say $G$ is unambiguous, reach-unambiguous or strongly unambiguous, respectively.
Furthermore, a family of directed graphs $\{G_x\}_{x\in X}$ is called \emph{few-unambiguous}, \emph{reach-few} or \emph{strongly-few} if there exists a polynomial $p:\mathbb{N} \rightarrow \mathbb{N}$ such that each of the graphs~$G_x$ from the family with~$|V(G_x)|=n$ nodes is $p(n)$-unambiguous, $p(n)$-reach-unambiguous or $p(n)$-strongly unambiguous, respectively.
\end{definition}

Consider the following examples of the above definition due to Lange \cite{Lan97}:
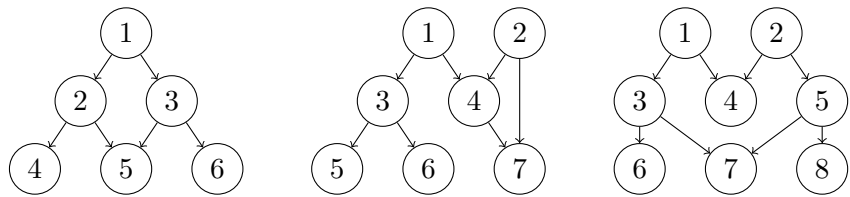
\begin{figure}[H]
    \centering
    \begin{tikzpicture}[scale=0.6]
        \node[shape=circle,draw=black] (1) at (0,0) {$1$};
        \node[shape=circle,draw=black] (2) at (-1,-1.5) {$2$};
        \node[shape=circle,draw=black] (3) at (1,-1.5) {$3$};
        \node[shape=circle,draw=black] (4) at (-2,-3) {$4$};
        \node[shape=circle,draw=black] (5) at (0,-3) {$5$};
        \node[shape=circle,draw=black] (6) at (2,-3) {$6$};
    
        \path [->](1) edge node[left] {} (2);
        \path [->](1) edge node[left] {} (3);
        \path [->](2) edge node[left] {} (4);
        \path [->](2) edge node[left] {} (5);
        \path [->](3) edge node[left] {} (5);
        \path [->](3) edge node[left] {} (6);
    \end{tikzpicture}
    \qquad
    \begin{tikzpicture}[scale=0.6]
        \node[shape=circle,draw=black] (1) at (0,0) {$1$};
        \node[shape=circle,draw=black] (2) at (2,0) {$2$};
        \node[shape=circle,draw=black] (3) at (-1,-1.5) {$3$};
        \node[shape=circle,draw=black] (4) at (1,-1.5) {$4$};
        \node[shape=circle,draw=black] (5) at (-2,-3) {$5$};
        \node[shape=circle,draw=black] (6) at (0,-3) {$6$};
        \node[shape=circle,draw=black] (7) at (2,-3) {$7$};
    
        \path [->](1) edge node[left] {} (3);
        \path [->](1) edge node[left] {} (4);
        \path [->](2) edge node[left] {} (4);
        \path [->](2) edge node[left] {} (7);
        \path [->](3) edge node[left] {} (5);
        \path [->](3) edge node[left] {} (6);
        \path [->](4) edge node[left] {} (7);
    \end{tikzpicture}
    \qquad
    \begin{tikzpicture}[scale=0.6]
        \node[shape=circle,draw=black] (1) at (0,0) {$1$};
        \node[shape=circle,draw=black] (2) at (2,0) {$2$};
        \node[shape=circle,draw=black] (3) at (-1,-1.5) {$3$};
        \node[shape=circle,draw=black] (4) at (1,-1.5) {$4$};
        \node[shape=circle,draw=black] (5) at (3,-1.5) {$5$};
        \node[shape=circle,draw=black] (6) at (-1,-3) {$6$};
        \node[shape=circle,draw=black] (7) at (1,-3) {$7$};
        \node[shape=circle,draw=black] (8) at (3,-3) {$8$};
    
        \path [->](1) edge node[left] {} (3);
        \path [->](1) edge node[left] {} (4);
        \path [->](2) edge node[left] {} (4);
        \path [->](2) edge node[left] {} (5);
        \path [->](3) edge node[left] {} (6);
        \path [->](3) edge node[left] {} (7);
        \path [->](5) edge node[left] {} (7);
        \path [->](5) edge node[left] {} (8);
    \end{tikzpicture}
    \label{unambiguity}
\end{figure}
\noindent The left graph is unambiguous with respect to nodes $1$ and $6$ but not reach-unambiguous with respect to node $1$, the middle one is reach-unambiguous with respect to node $1$ but not strongly unambiguous and the right one is strongly unambiguous.

These notions of unambiguity and fewness naturally give rise to six complexity classes between~$\mathsf{L}$ and $\mathsf{NL}$.
We follow the original paper of Buntrock, Jenner, Lange and Rossmanith \cite{BJLR91} and define
    \emph{strongly-unambiguous logspace}, $\mathsf{StrongUL}$, 
    \emph{strongly-few logspace}, $\mathsf{StrongFewL}$,
    \emph{reach-unambiguous logspace}, $\mathsf{ReachUL}$,
    \emph{reach-few logspace}, $\mathsf{ReachFewL}$,
    \emph{unambiguous logspace},~$\mathsf{UL}$, and
    \emph{few logspace},~$\mathsf{FewL}$ 
    as the classes of languages that are decidable by an~$\mathsf{NL}$-Turing machine $M$ with unique accepting configuration whose family of configuration graphs for inputs~$x~\in~\Sigma^*$,~$\{G_{M,x}\}_{x\in\Sigma^*}$, satisfies the corresponding unambiguity or fewness restriction with respect to its starting and its accepting configuration.

Recall that \cref{thrm: counting paths on StFewL-graphs} implies the inclusion:
\stfewl*

Putting this together with the trivial containments and the previously mentioned results in~\cite{BJLR91,Lan97,AL98,GSTV11}, we obtain the following inclusion diagram to read from left to right:
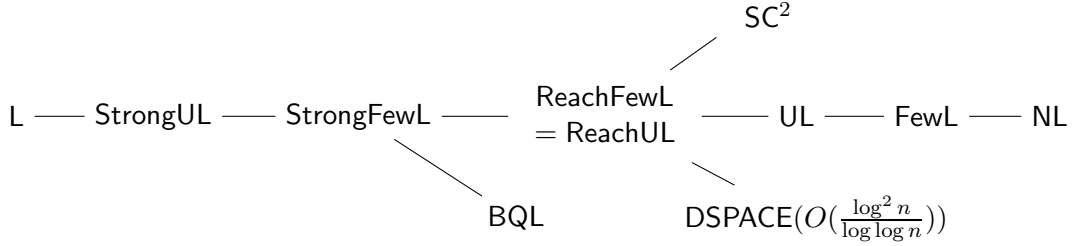
\begin{figure}[H]
\begin{center}
\begin{tikzpicture}[scale=0.85]
    % nodes:
    \node[] (1) at (0,0) {$\mathsf{L}$};
    \node[] (2) at (2,0) {$\mathsf{StrongUL}$};
    \node[] (3) at (5,0) {$\mathsf{StFewL}$};
    \node[] (4) at (8.6,0) 
        {\begin{tabular}{c} $\mathsf{ReachFewL}$ \\ $=\mathsf{ReachUL}$ \end{tabular}};
    \node[] (5) at (11.4,0) {$\mathsf{UL}$};
    \node[] (6) at (13.3,0) {$\mathsf{FewL}$};
    \node[] (7) at (15.1,0) {$\mathsf{NL}$};
    \node[] (8) at (7.3,-1.5) {$\mathsf{BQL}$};
    \node[] (9) at (11.4,-1.5) 
        {\text{ }\text{ }\text{ }\text{ }$\mathsf{DSPACE}(O(\frac{\log^2 n}{\log \log n}))$};
    \node[] (10) at (10.7,1.5) 
        {\text{ }\text{ }\text{ }\text{ }$\mathsf{SC}^2$};

    % edges:
    \path [-](1) edge node[left] {} (2);
    \path [-](2) edge node[left] {} (3);
    \path [-](3) edge node[left] {} (4);
    \path [-](4) edge node[left] {} (5);
    \path [-](5) edge node[left] {} (6);
    \path [-](6) edge node[left] {} (7);
    \path [-](3) edge node[left] {} (8);
    \path [-](4) edge node[left] {} (9);
    \path [-](4) edge node[left] {} (10);
\end{tikzpicture}
\end{center}
\label{unambiguity and fewness complexity diagram}
\end{figure}

\subsection{Proof of Theorem 4} \label{sec:BQL-language}

Observe that \cref{thrm: counting paths on StFewL-graphs} only shows that we can decide $st$-connectivity on directed graphs that are promised to be strongly-few. We now show that we can also check whether this promise holds in~$\mathsf{BQSPACE}(O(\log n))$, i.e. we obtain a language containment. We restate the result mentioned in the introduction:
\STCONsf*

The idea for proving the above theorem is to use Ta-Shma's spectrum approximation procedure to estimate the smallest singular value of the counting Laplacian $L$.
In case the graph is acyclic and the smallest singular value is below some threshold, this gives us a lower bound for~$\|L^{-1}\|_{\max}$, the maximum number of paths, so that we can correctly reject graphs with too many paths.
In case it is higher than the threshold, we know that $L$ is poly-conditioned and we can proceed as in \cref{thrm: counting paths on StFewL-graphs} to compute the number of paths between any pair of nodes exactly.
However, in order for this approach to work, it is crucial that the graph is acyclic. Otherwise, the smallest singular value of $L$ need not be related to the number of paths at all.
We ensure this by first mapping the input graph to a layered graph that is guaranteed to be acyclic, similar as in \cite{GSTV11}.
We give the following definition.

\begin{definition}
    Let $G=(V,E)$ be a directed graph with $|V|=n$ vertices. 
    We define the layered graph $\mathrm{lay}(G)$ on the vertex set $V':=V\times\{0,1,...,n\}$ with two types of edges:
    \begin{enumerate}
        \item For all edges $(i,j)\in E$ and all $l\leq n-2$ add an edge from $(i,l)$ to $(j,l+1)$ in $\mathrm{lay}(G)$,
        \item for all $i \in V$ and all $l\leq n-1$ add an edge from $(i,l)$ to $(i,n)$ in $\mathrm{lay}(G)$.
    \end{enumerate}
\end{definition}

It is easy to see that the paths in the first $n$ layers of $\mathrm{lay}(G)$ directly correspond to paths of length less than $n$ in $G$.
The last layer in $\mathrm{lay}(G)$ just serves as a catch basin for all paths of different lengths.
Let us quickly remark that counting paths in $\mathrm{lay}(G)$ also allows us to detect cycles in $G$. To see this, note that there is a cycle in $G$ through a node $i$ if and only if~$N((i,0),(i,n))\geq 2$ in $\mathrm{lay}(G)$.
With this in mind, let us now prove the above theorem.
\begin{proof}[Proof of \cref{STCONsf}]
    Let $\braket{G,s,t,1^k}$ be a given input graph instance with $|V(G)|=n$ nodes.
    Without loss of generality assume that $k\leq p(n)$ for some fixed polynomial $p:\mathbb{N}\rightarrow\mathbb{N}$. 
    We first construct $\mathrm{lay}(G)$ which is acyclic. Note that this is possible in $\mathsf{AC}^0$.
    We consider the counting Laplacian~$L$ of~$\mathrm{lay}(G)$ and run Ta-Shma's spectrum approximation algorithm (compare \cite[Theorem 5.2]{TS13}) to approximate its singular values with error $\varepsilon=1/6$ and accuracy $\delta = 1/(2n\cdot k)$.
    If we obtain a singular value smaller than $\delta$, then we have with probability at least $5/6$,
    \begin{align*}
        2 \delta 
        = \frac{1}{n\cdot k} 
        > \sigma_n(L) 
        = \frac{1}{\|L^{-1}\|_2} 
        \geq \frac{1}{n\cdot \|L^{-1}\|_{\max}}
        = \frac{1}{n\cdot \max_{i,j\in V(\mathrm{lay}(G))} N(i,j)}
    \end{align*}
    which implies $\max_{i,j \in V(G)} N(i,j) \geq \max_{i,j\in V(\mathrm{lay}(G))} N(i,j) > k$.
    In this case we reject the input.
    Otherwise, if we do not obtain a singular value smaller than $\delta$, then we know with the same probability that the counting Laplacian of $\mathrm{lay}(G)$ is poly-conditioned.
    Thus, we can proceed as in \cref{thrm: counting paths on StFewL-graphs} and run Ta-Shma's matrix inversion algorithm to determine all entries of~$L^{-1}$ with total error~$\varepsilon'=1/6$ and accuracy~$1/3$.
    For all~$i\in V(G)$ we check in this way whether~$L^{-1}((i,0),(i,n)) \geq 2$. If this is the case for some~$i\in V(G)$, then there is a cycle in~$G$ containing $i$, i.e.~$N(i,i)=\infty$ in~$G$, and we reject the input.
    Otherwise, we further check whether all entries of $L^{-1}$ are upper bounded by~$k$ and if~$L^{-1}((s,0),(t,n)) \geq 1$.
    The former implies~$N(i,j)\leq k$ for all~$i,j\in V(G)$, and the latter implies~$N(s,t)\geq 1$ in $G$.
    If both conditions are satisfied, we accept the input.
    Otherwise, we reject it.
    The total error probability of the algorithm is no higher than~$\varepsilon + \varepsilon' = 1/3$.
\end{proof}

%%%%%%%%%%%%%%%%%%%% COUNTING PATHS: %%%%%%%%%%%%%%%%%%%%

\section{Counting few paths in quantum logspace}
\label{sec:4}

In this section we show how to push the algorithmic idea behind \cref{thrm: counting paths on StFewL-graphs} to obtain \cref{main}, which shows how to count $st$-paths on graphs with a polynomial bound on the number of paths leaving $s$, and on the number of paths arriving in $t$, in $\mathsf{BQSPACE}(O(\log n))$.
For this, we need the notion of an effective pseudoinverse, which we introduce in \cref{sec:pseudo}.
The proof of \cref{main} is in \cref{sec:proof-main}.

\subsection{Effective pseudoinverse} \label{sec:pseudo}

A close look at Ta-Shma's matrix inversion algorithm shows that it can be easily altered to also handle ill-conditioned matrices as input and only invert them on their well-conditioned part.
In order to appropriately state this observation we make the following definition.

\begin{definition}[Effective pseudoinverse]
    Let $M$ be an $n\times n$ matrix with singular value decomposition $M=\sum_{j=1}^n \sigma_j \ket{u_j} \bra{v_j}$. For $\zeta>0$ we define the \emph{$\zeta$-effective pseudoinverse} of $M$ as the matrix
    \begin{align*}
        M_\zeta^+ := \sum_{\sigma_j \geq \zeta} \sigma_j^{-1} \ket{v_j} \bra{u_j}.
    \end{align*}
\end{definition}

Note that in the definition we essentially drop the largest terms of the actual inverse $M^{-1}=\sum_{j=1}^n~\sigma_j^{-1}~\ket{v_j}~\bra{u_j}$.
While this produces significant error to approximate $M^{-1}$ as a whole, we find that it can still give a good approximation for some relevant entries. We now state the refined version of Ta-Shma's matrix inversion algorithm which allows us to compute effective pseudoinverses of general matrices.

\begin{theorem}
    Fix $\varepsilon(n), \zeta(n), \delta(n) > 0$ and $Z(n) \geq 1$.
    Let $M$ be an $n \times n$ matrix such that $Z \geq \sigma_1(M) \geq ... \geq \sigma_n(M)$.
    There is an algorithm running in $\mathsf{BQSPACE}(O(\log \frac{nZ}{\varepsilon\delta}))$ that given~$M$ and two indices $s,t\in [n]$ outputs with probability $1-\varepsilon$ an $\varepsilon$-additive approximation of the entry $M_{\widetilde{\zeta}}^+(s,t)$,
    where $\widetilde{\zeta}$ is a random value that is fixed at the beginning of the computation and is $\delta$-close to $\zeta$.
    \label{pseudoinverse}
\end{theorem}

While it is a simple modification of \cite{TS13}, for completeness we provide a proof of this theorem in \cref{app:ta-shma-pseud}.
The fact that we cannot control the exact threshold $\zeta$ of which singular values should be ignored during the effective pseudoinversion is a consequence of the inaccuracy of quantum phase estimation. This inaccuracy is limited by using Ta-Shma's consistent phase estimation procedure.

\subsection{Proof of Theorem 2} \label{sec:proof-main}
We now have the necessary tools to prove our final result which we recall here:
\main*
The idea for the proof is to approximate the effective pseudoinverse entry of the counting Laplacian $L^+_{\Tilde{\zeta}}(s,t)$ for some small enough $\Tilde{\zeta} = 1/\mathrm{poly}(n)$ instead of the actual inverse entry~$L^{-1}(s,t)$.
While ignoring the smallest singular values during the effective pseudoinversion normally leaves out the largest terms, we find that our path bounds imply low overlap of~$\braket{s|v_j}$ and~$\braket{u_j|t}$ for small singular values $\sigma_j$ such that entry~$L^+_{\Tilde{\zeta}}(s,t)$ is close to~$L^{-1}(s,t)$.
\begin{proof}
    Without loss of generality we assume the graph to be acyclic. Otherwise, we proceed as in \cref{sec:BQL-language} and first map it to $\mathrm{lay}(G)$ which is guaranteed to be acyclic.
    Now consider the counting Laplacian and its singular value decomposition $L=I-A = \sum_{j=1}^n \sigma_j \ket{u_j}\bra{v_j}$. Its inverse is $L^{-1} = \sum_{j=1}^n \sigma_j^{-1} \ket{v_j}\bra{u_j}$.
    Further, consider the vectors~$(L^{-1})^T \ket{s}$ and $L^{-1}\ket{t}$.
    They contain as entries the number of paths starting in $s$ and the number of paths ending in~$t$, respectively. Hence, we find for their squared $\ell_2$-norms:
    \begin{align*}
        \left\|(L^{-1})^T \ket{s}\right\|_2^2 
        &= \left\| \sum_{j=1}^n \sigma_j^{-1} \ket{u_j} \braket{v_j|s} \right\|_2^2
        = \sum_{j=1}^n \sigma_j^{-2} \left|\braket{v_j|s}\right|^2
        \leq n \cdot p(n)^2 \qquad \text{ and similarly}\\
        \left\|L^{-1}\ket{t}\right\|_2^2 
        &= \left\| \sum_{j=1}^n \sigma_j^{-1} \ket{v_j}\braket{u_j|t} \right\|_2^2
        = \sum_{j=1}^n \sigma_j^{-2} \left|\braket{u_j|t}\right|^2
        \leq n \cdot p(n)^2
    \end{align*}
    where the last equalities of both lines follow because the $\{\ket{u_j}\}_{j\in[n]}$ and $\{\ket{v_j}\}_{j\in[n]}$ are orthonormal bases.
    As a consequence we obtain for each $j \in [n]$
    \begin{align*}
        \left|\braket{s|v_j}\right| \leq \sigma_j \sqrt{n} \cdot p(n) 
        \quad \text{and} \quad
        \left|\braket{u_j|t}\right| \leq \sigma_j \sqrt{n} \cdot p(n).
    \end{align*}
    Combining the two we find as a bound for the $j$-th term of the singular value decomposition of~$L^{-1}$,
    \begin{align*}
        \sigma_j^{-1} \left|\braket{s|v_j} \braket{u_j|t}\right| \leq \sigma_j n \cdot p(n)^2.
    \end{align*}
    This allows us to estimate the error in computing an effective pseudoinverse entry $L_{\widetilde{\zeta}}^+(s,t)=\sum_{\sigma_j \geq \widetilde{\zeta}} \sigma_j^{-1} \braket{s|v_j} \braket{u_j|t}$ instead of the actual inverse entry $L^{-1}(s,t) = N(s,t)$.
    In fact, for~$\widetilde{\zeta}\leq\frac{1}{5n^2\cdot p(n)^2}$ we have
    \begin{align*}
        \left| L^{-1}(s,t) - L_{\widetilde{\zeta}}^+(s,t) \right|
        \leq \sum_{\sigma_j < \widetilde{\zeta}} \sigma_j^{-1} \left| \braket{s|v_j} \braket{u_j|t} \right| 
        < \widetilde{\zeta} n^2 \cdot p(n)^2
        \leq 1/5.
    \end{align*}
    Choosing $Z=n$, $\varepsilon=1/5$ and $\delta = \zeta = \frac{1}{10n^2\cdot p(n)^2}$ in \cref{pseudoinverse} ensures~$\widetilde{\zeta}~\leq~\frac{1}{5n^2~\cdot~p(n)^2}$ and yields an additive $2/5$ approximation of $L^{-1}(s,t)$ within the desired space complexity. 
    Rounding to the closest integer gives the number of paths from $s$ to $t$.
\end{proof}

A natural open question is whether the simultaneous polynomial bound on (i) the number of paths starting from $s$ and on (ii) the number of paths ending in $t$ is really necessary for a~$\mathsf{BQSPACE}(O(\log n))$ procedure.
Unfortunately, at least with our approach above, this seems to be the case.
Note that (i) implies low overlap of~$\ket{s}$ with the left singular vectors~$\ket{v_j}$ of $L^{-1}$ and (ii) implies low overlap of $\ket{t}$ with the right singular vectors $\ket{u_j}$ of $L^{-1}$. 
It turns out that if only one of the two is small, then the contribution of $\sigma_j^{-1} \braket{s|v_j} \braket{u_j|t}$ to~$L^{-1}(s,t)$ can still be significant for very small $\sigma_j$ but will be ignored in the pseudoinversion.

\section*{Acknowledgements}
We thank Fran{\c{c}}ois Le Gall for discussions about the differences between language and promise classes. He made us aware that our results can be seen as first evidence that the language classes $\BQL$ and $\BPL$ are distinct.
We further acknowledge useful discussions with Klaus-Jörn Lange and we thank Lior Eldar, Troy Lee and Ronald de Wolf for valuable comments on an earlier draft.
This work has received support under the program ``Investissement d'Avenir'' launched by the French Government and implemented by ANR, with the reference ``ANR‐22‐CMAS-0001, QuanTEdu-France''.
SA was partially supported by French projects EPIQ (ANR-22-PETQ-0007), QUDATA (ANR18-CE47-0010), QUOPS (ANR-22-CE47-0003-01) and HQI (ANR-22-PNCQ-0002), and EU project QOPT (QuantERA ERA-NET Cofund 2022-25).

\bibliographystyle{alpha}
\bibliography{refs}

\newcommand{\etalchar}[1]{$^{#1}$}
\begin{thebibliography}{AKL{\etalchar{+}}79}

\bibitem[AB09]{AB06}
Sanjeev Arora and Boaz Barak.
\newblock {\em Computational Complexity: A Modern Approach}.
\newblock Cambridge University Press, USA, 1st edition, 2009.

\bibitem[AKL{\etalchar{+}}79]{AKL+79}
Romas Aleliunas, Richard~M. Karp, Richard~J. Lipton, Laszlo Lovasz, and Charles Rackoff.
\newblock Random walks, universal traversal sequences, and the complexity of maze problems.
\newblock In {\em 20th Annual Symposium on Foundations of Computer Science}, pages 218--223, 1979.

\bibitem[AL98]{AL98}
Eric Allender and klaus-joern Lange.
\newblock $\mathsf{RUSPACE}(\log n) \subseteq \mathsf{DSPACE}(\log^2 n/\log \log n)$.
\newblock {\em Theory of Computing Systems}, 31, 10 1998.

\bibitem[All23]{All23}
Eric Allender.
\newblock Guest column: Parting thoughts and parting shots (read on for details on how to win valuable prizes!
\newblock {\em SIGACT News}, 54(1):63–81, March 2023.

\bibitem[BBRS98]{BBRS98}
Greg Barnes, Jonathan~F. Buss, Walter~L. Ruzzo, and Baruch Schieber.
\newblock A sublinear space, polynomial time algorithm for directed s-t connectivity.
\newblock {\em SIAM Journal on Computing}, 27(5):1273--1282, 1998.

\bibitem[BJLR91]{BJLR91}
Gerhard Buntrock, Birgit Jenner, Klaus-J\"{o}rn Lange, and Peter Rossmanith.
\newblock Unambiguity and fewness for logarithmic space.
\newblock In {\em Proceedings of the 8th International Symposium on Fundamentals of Computation Theory}, FCT '91, page 168–179, Berlin, Heidelberg, 1991. Springer-Verlag.

\bibitem[Coo85]{Coo85}
Stephen~A. Cook.
\newblock A taxonomy of problems with fast parallel algorithms.
\newblock {\em Information and Control}, 64(1):2--22, 1985.
\newblock International Conference on Foundations of Computation Theory.

\bibitem[DSTS17]{DST17}
Dean Doron, Amir Sarid, and Amnon Ta-Shma.
\newblock On approximating the eigenvalues of stochastic matrices in probabilistic logspace.
\newblock {\em Comput. Complex.}, 26(2):393–420, June 2017.

\bibitem[FL18]{FL18}
Bill Fefferman and Cedric Yen-Yu Lin.
\newblock {A Complete Characterization of Unitary Quantum Space}.
\newblock In Anna~R. Karlin, editor, {\em 9th Innovations in Theoretical Computer Science Conference (ITCS 2018)}, volume~94 of {\em Leibniz International Proceedings in Informatics (LIPIcs)}, pages 4:1--4:21, Dagstuhl, Germany, 2018. Schloss Dagstuhl -- Leibniz-Zentrum f{\"u}r Informatik.

\bibitem[FR21]{FR21}
Bill Fefferman and Zachary Remscrim.
\newblock Eliminating intermediate measurements in space-bounded quantum computation.
\newblock In {\em Proceedings of the 53rd Annual ACM SIGACT Symposium on Theory of Computing}, STOC ’21, page 1343–1356. ACM, June 2021.

\bibitem[GLW24]{GLW23}
François~Le Gall, Yupan Liu, and Qisheng Wang.
\newblock Space-bounded quantum state testing via space-efficient quantum singular value transformation, 2024.

\bibitem[GR21]{GR21}
Uma Girish and Ran Raz.
\newblock Eliminating intermediate measurements using pseudorandom generators, 2021.

\bibitem[GRZ21]{GRZ21}
Uma Girish, Ran Raz, and Wei Zhan.
\newblock {Quantum Logspace Algorithm for Powering Matrices with Bounded Norm}.
\newblock In Nikhil Bansal, Emanuela Merelli, and James Worrell, editors, {\em 48th International Colloquium on Automata, Languages, and Programming (ICALP 2021)}, volume 198 of {\em Leibniz International Proceedings in Informatics (LIPIcs)}, pages 73:1--73:20, Dagstuhl, Germany, 2021. Schloss Dagstuhl -- Leibniz-Zentrum f{\"u}r Informatik.

\bibitem[GSTV11]{GSTV11}
Brady Garvin, Derrick Stolee, Raghunath Tewari, and N.~Vinodchandran.
\newblock $\mathrm{ReachFewL} = \mathrm{ReachUL}$.
\newblock {\em Electronic Colloquium on Computational Complexity (ECCC)}, 18:60, 08 2011.

\bibitem[HHL09]{HHL09}
Aram~W. Harrow, Avinatan Hassidim, and Seth Lloyd.
\newblock Quantum algorithm for linear systems of equations.
\newblock {\em Physical Review Letters}, 103(15), October 2009.

\bibitem[Hoz21]{Hoz21}
William~M. Hoza.
\newblock {Better Pseudodistributions and Derandomization for Space-Bounded Computation}.
\newblock In Mary Wootters and Laura Sanit\`{a}, editors, {\em Approximation, Randomization, and Combinatorial Optimization. Algorithms and Techniques (APPROX/RANDOM 2021)}, volume 207 of {\em Leibniz International Proceedings in Informatics (LIPIcs)}, pages 28:1--28:23, Dagstuhl, Germany, 2021. Schloss Dagstuhl -- Leibniz-Zentrum f{\"u}r Informatik.

\bibitem[KKR08]{KKR08}
Sampath Kannan, Sanjeev Khanna, and Sudeepa Roy.
\newblock Stcon in directed unique-path graphs.
\newblock In {\em FSTTCS}, pages 256--267, 2008.

\bibitem[Lan97]{Lan97}
Klaus-J\"{o}rn Lange.
\newblock An unambiguous class possessing a complete set.
\newblock In {\em Proceedings of the 14th Annual Symposium on Theoretical Aspects of Computer Science}, STACS '97, page 339–350, Berlin, Heidelberg, 1997. Springer-Verlag.

\bibitem[LPW08]{LPW08}
D.A. Levin, Y.~Peres, and E.L. Wilmer.
\newblock {\em Markov Chains and Mixing Times}.
\newblock American Mathematical Soc., 2008.

\bibitem[MW12]{MW10}
Dieter~van Melkebeek and Thomas Watson.
\newblock Time-space efficient simulations of quantum computations.
\newblock {\em Theory of Computing}, 8(1):1--51, 2012.

\bibitem[Rei08]{Rei05}
Omer Reingold.
\newblock Undirected connectivity in log-space.
\newblock {\em J. ACM}, 55(4), September 2008.

\bibitem[Sav70]{Sav70}
Walter~J. Savitch.
\newblock Relationships between nondeterministic and deterministic tape complexities.
\newblock {\em Journal of Computer and System Sciences}, 4(2):177--192, 1970.

\bibitem[SZ99]{SZ99}
Michael Saks and Shiyu Zhou.
\newblock $\mathsf{BPHSPACE}(s)\subseteq\mathsf{DSPACE}(s^{3/2})$.
\newblock {\em Journal of Computer and System Sciences}, 58(2):376--403, 1999.

\bibitem[TS13]{TS13}
Amnon Ta-Shma.
\newblock Inverting well conditioned matrices in quantum logspace.
\newblock In {\em Proceedings of the Forty-Fifth Annual ACM Symposium on Theory of Computing}, STOC '13, page 881–890, New York, NY, USA, 2013. Association for Computing Machinery.

\end{thebibliography}

\appendix

\section{Effective pseudoinversion in quantum logspace} \label{app:ta-shma-pseud}
The algorithm for the effective pseudoinversion is essentially identical to Ta-Shma's matrix inversion procedure except for a small change in the rotation step to appropriately treat ill-conditioned matrices. 
For completeness we describe it here but refer to \cite{TS13} for the detailed complexity and error analysis.
\begin{proof}[Proof sketch of \cref{pseudoinverse}.]
    We make the following two simplifying assumptions:
    \begin{itemize}
        \item Without loss of generality we assume that $Z=1$. Otherwise, we simply choose $Z~=~n\cdot~\|M\|_{\max}$ as an upper bound for $\sigma_1(M)$ and rescale the matrix with factor $1/Z$.
            Approximating the entries of a $\tilde{\big(\frac{\zeta}{Z}\big)}$-effective pseudoinverse of this rescaled matrix up to accuracy $\varepsilon \cdot Z$ gives an~$\varepsilon$ approximation of the entries of a $\tilde{\zeta}$-effective pseudoinverse of $M$ within the same space complexity.
        \item Furthermore, we assume the matrix $M$ to be Hermitian. Otherwise we use the well-known reduction to the Hermitian case
            \begin{align*}
                H = H(M) := \begin{bmatrix}
                    0 & M^{\dag} \\
                    M & 0
                \end{bmatrix}
                \qquad \text{ with inverse}
                \qquad
                H^{-1} = \begin{bmatrix}
                    0 & M^{-1} \\
                    (M^\dag)^{-1} & 0
                \end{bmatrix}.
            \end{align*}
            It is easily verified that the eigenpairs of $H$ are given by $\{\pm\sigma_j, \frac{1}{\sqrt{2}} \left(\ket{0}\ket{v_j} \pm \ket{1} \ket{u_j}\right)\}_{j \in [n]}$, where $\sigma_j$ and $\ket{v_j},\ket{u_j}$ denote the singular values and vectors of $M$, respectively. We then find $M^+_{\widetilde{\zeta}}(s,t) = H^+_{\widetilde{\zeta}}(s,t+n)$.
    \end{itemize}

    \noindent Assuming the above, let the spectral decomposition of the matrix be given by $M=H=\sum_{j=1}^n \lambda_j \ket{h_j}\bra{h_j}$ and let $\ket{t} = \sum_{j=1}^n \beta_j \ket{h_j}$.
    The algorithm works on four registers: An input register $I$, an estimation register $E$, a shift register $S$ and an ancillary register $A$ of dimension at least three.
    \begin{enumerate}
        \item We start by preparing the initial state
            \begin{equation*}
                \ket{t}_I\ket{0}_E\ket{0}_S\ket{\text{initial}}_A
                = \sum_{j=1}^n \beta_j \ket{h_j}_I\ket{0}_E\ket{0}_S\ket{\text{initial}}_A.
            \end{equation*}
        \item We then apply the consistent phase estimation procedure (compare \cite[Section 5.2]{TS13}) to the input, estimation and shift register and obtain a state close to
            \begin{align*}
                \sum_{j=1}^n \beta_j \ket{h_j}_I \ket{0}_E \ket{s(j)}_S
            \end{align*}
            where $s(j)$ is the \emph{j-th section number}, that is a fixed classical value depending on $h_j$ only, from which we can recover a $\delta$-approximation $\widetilde{\lambda}_j = \widetilde{\lambda}(s(j))$ of the eigenvalue $\lambda_j$.
        \item We next approximately apply a unitary map acting on the shift and ancillary register partially described by
            \begin{align*}
                \ket{s}_S \ket{\text{initial}}_A \mapsto 
                    \begin{cases}
                        \ket{s}_S \left( \frac{\zeta}{\widetilde{\lambda}(s)} \ket{\text{well}}_A + \sqrt{1 - \left( \frac{\zeta}{\widetilde{\lambda}(s)} \right)^2} \ket{\text{nothing}}_A \right)
                            & \text{ if } \left|\widetilde{\lambda}(s)\right| \geq \zeta,\\
                        \ket{s}_S \ket{\text{ill}}_A
                            & \text{ if } \left|\widetilde{\lambda}(s)\right| < \zeta.
                    \end{cases}
            \end{align*}
        \item We reverse the consistent phase estimation and are left with a state close to
            \begin{equation*}
                \sum_{\left|\widetilde{\lambda}_j\right|\geq \zeta}^n \beta_j \ket{h_j}_I \left(
                    \frac{\zeta}{\widetilde{\lambda}_j} \ket{\text{well}}_A + \sqrt{1-\left( \frac{\zeta}{\widetilde{\lambda}_j}\right)^2} \ket{\text{nothing}}_A\right)
                    + \sum_{\left|\widetilde{\lambda}_j \right| < \zeta}^n \beta_j \ket{h_j}_I \ket{\text{ill}}_A.
            \end{equation*}
            Note that the approximations of the eigenvalues are monotone in the sense that $\lambda_i \leq \lambda_j$ implies $\widetilde{\lambda}_i\leq\widetilde{\lambda}_j$.
            From this we get the existence of $\widetilde{\zeta}_+ > 0$ and $\widetilde{\zeta}_- < 0$, both in absolute value $\delta$-close to $\zeta$, such that
                $\{j\in[n]:|\widetilde{\lambda}_j|\geq\zeta\} 
                = \{j\in[n]:\lambda_j\geq\tilde{\zeta}_{+} \text{ or }\lambda_j\leq\tilde{\zeta}_{-}\}$.
            Choosing and combining symmetric shifts for the positive and negative eigenvalues during the consistent phase estimation allows to assume $|\widetilde{\zeta}_+| = |\widetilde{\zeta}_-|$.
            Denoting this value by $\tilde{\zeta}$ we find
                $\{j\in[n]:|\widetilde{\lambda}_j|\geq\zeta\} = \{j\in[n]:|\lambda_j|\geq\tilde{\zeta}\}$.
        \item Finally, we measure the ancillary register. If the measurement outcome is $\ket{\text{well}}_A$, this leaves us with a state close to the normalized desired one $\frac{1}{\big\|M_{\widetilde{\zeta}}^+\ket{t}\big\|_2} M_{\widetilde{\zeta}}^+\ket{t}$.
            In fact, estimating the success-probability of this measurement outcome gives a good approximation of
            \begin{equation*}
                \zeta^2 \sum_{|\lambda_j|\geq\widetilde{\zeta}} \beta_j^2 \lambda_j^{-2} = \zeta^2 \Big\|M_{\widetilde{\zeta}}^+\ket{t}\Big\|_2^2
            \end{equation*}
            from which we recover the norm $\Big\|M_{\widetilde{\zeta}}^+\ket{t}\Big\|_2$.
    \end{enumerate}
    
    Repeating the steps above sufficiently many times lets us create multiple copies of the final state close to $\frac{1}{\big\|M_{\widetilde{\zeta}}^+\ket{t}\big\|_2} M_{\widetilde{\zeta}}^+\ket{t}$.
    We then use Ta-Shma's space-efficient tomography procedure \cite[Theorem 6.1]{TS13} to approximately learn the state, and in particular entry~$\frac{1}{\big\|M_{\widetilde{\zeta}}^+\ket{t}\big\|_2} \bra{s} M_{\widetilde{\zeta}}^+\ket{t}$.
\end{proof}

\end{document}